\documentclass{article}
\usepackage{ijcai26}

\usepackage{times}
\usepackage{soul}
\usepackage{url}
\usepackage[hidelinks]{hyperref}
\usepackage[utf8]{inputenc}
\usepackage[small]{caption}
\usepackage{graphicx}
\usepackage{amsmath}
\usepackage{amsthm}
\usepackage{booktabs}
\usepackage[switch]{lineno}

\usepackage{amssymb}
\usepackage{algorithm}
\usepackage[noend]{algpseudocode}
\algnewcommand\algorithmicinput{\textbf{Input:}}
\algnewcommand\AlgInput{\item[\algorithmicinput]}
\algnewcommand\algorithmicoutput{\textbf{Output:}}
\algnewcommand\AlgOutput{\item[\algorithmicoutput]}
\algnewcommand\algorithmicforeach{\textbf{for each}}
\algdef{S}[FOR]{ForEach}[1]{\algorithmicforeach\ #1\ \algorithmicdo}
\algrenewcommand\algorithmicindent{1em}
\usepackage{tikz}
\usetikzlibrary{calc}
\usepackage{multirow}
\usepackage{xurl}    
\usepackage{pifont}
\usepackage{bm}
\usepackage[capitalize]{cleveref}
\usepackage{siunitx}
\usepackage{color}

\newcommand{\noshow}[1]{}

\newtheorem{theorem}{Theorem}
\newtheorem{lemma}{Lemma}

\Crefname{figure}{Figure}{Figures}
\crefname{figure}{Fig.}{Fig.}
\crefname{theorem}{Thrm.}{Thrm.}
\Crefname{theorem}{Theorem}{Theorems}
\Crefname{algorithm}{Algorithm}{Algorithms}
\crefname{algorithm}{Alg.}{Alg.}
\crefname{line}{Line}{Lines}
\crefname{table}{Table}{Tables}
\crefname{section}{Sec.}{Sec.}
\Crefname{section}{Section}{Sections}
\crefname{definition}{Def.}{Def.}
\Crefname{definition}{Definition}{Definitions}
\crefname{lemma}{Lemma}{Lemmas}
\Crefname{lemma}{Lemma}{Lemmas}
\crefname{proposition}{Prop.}{Prop.}
\Crefname{proposition}{Proposition}{Propositions}
\crefname{observation}{Observation}{Observations}
\Crefname{observation}{Observation}{Observations}
\crefname{equation}{Eq.}{Eq.}
\Crefname{equation}{Equation}{Equations}
\Crefname{part}{Part}{Parts}
\crefname{part}{Part}{Parts}
\Crefname{chapter}{Chapter}{Chapters}
\crefname{chapter}{Chap.}{Chap.}
\crefname{appendix}{Appendix}{Appendices}
\Crefname{appendix}{Appendix}{Appendices}

\title{Conveyor Parcel Routing with Order-Contiguous Arrivals}

\author{
Takuro Kato$^1$
\And
Keisuke Okumura$^2$\\
\affiliations
$^1$Toyota Industries Corporation, Japan\\
$^2$National Institute of Advanced Industrial Science and Technology (AIST), Japan\\
\emails
takuro.kato@mail.toyota-shokki.co.jp,
okumura.k@aist.go.jp
}

\begin{document}

\maketitle

\begin{abstract}
In warehouse logistics, parcels released from the outfeed of an automated storage system must be routed through conveyor networks to workstations. 
Beyond collision avoidance, practical operations impose an additional requirement of \emph{order-contiguous arrivals}: at each delivery point, parcels belonging to the same order must arrive as a consecutive block in the arrival sequence to reduce downstream re-sorting effort. 
We formalize this problem as \emph{online multi-agent path finding with order-contiguity (online MAPF-OC)}, where agents (i.e., parcels) appear over time and exit upon delivery. 
To efficiently solve online MAPF-OC, we propose \emph{Dual-Ordering Prioritized Planning (DOPP)}, a complete polynomial-time algorithm with a three-level structure that \emph{(i)} searches order-level arrival sequences, \emph{(ii)} refines agent-level priorities, and \emph{(iii)} synthesizes feasible solutions via prioritized planning. 
Experiments on various conveyor-network layouts, including those derived from actual warehouses, demonstrate DOPP's practical scalability and ability to generate high-quality plans within tight time budgets.
\end{abstract}

\section{Introduction}\label{sec:introduction}
{
\begin{figure}[t]
  \centering
  \includegraphics{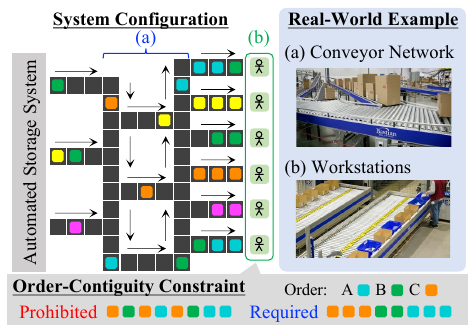}
  \caption{Conveyor-based fulfillment system in warehouse logistics. Parcels retrieved at the outfeed of an automated storage system (e.g., grid-based robotic storage systems~\protect\cite{salzman20research,chen21integrated}) enter a directed one-way conveyor network and are routed online to workstations under collision avoidance. The \emph{order-contiguity} constraint requires parcels of the same order to arrive contiguously (i.e., without interleaving by other orders) at each workstation. Photos sourced from \protect\cite{bastian_examples}.
  }
  \label{fig:overview}
\end{figure}
}
In modern fulfillment centers, parcels retrieved at the outfeed of an automated storage system are injected into the conveyor network and delivered to workstations.
\Cref{fig:overview} illustrates a typical setup.
While conveyors are a standard component of intra-facility transport, networks with many merges and splits are prone to congestion and blocking.
Thus, even if the upstream storage system releases parcels quickly and the terminal workstations process parcels sufficiently fast, overall throughput can still be constrained by delays within the intermediate conveyor network.
Therefore, efficient outfeed-to-workstation routing is crucial for system performance.

Such a scenario can be naturally abstracted as \emph{multi-agent path finding (MAPF)}~\cite{stern19mapf}, which has been widely studied in warehouse automation to synthesize the collision-free motion of many robots performing upstream retrieval.
However, parcel routing on conveyor networks, which often form the downstream transport layer, has received comparatively less attention.
We model this transport as \emph{online MAPF}~\cite{svancara19online}, where agents (i.e., parcels) are released over time at the storage outfeed, traverse the directed conveyor network under collision avoidance, are processed at workstations by human/robotic operators, and then exit.

Beyond collision avoidance captured by online MAPF, practical fulfillment operations impose an additional requirement: parcels must be grouped by \emph{order} and handled together at the same workstation. 
If parcels of different orders interleave in the arrival sequence, operators must re-sort them, increasing workload and processing delay.
This poses a novel \emph{order-contiguity} constraint: at each workstation, parcels with the same order must arrive consecutively, without interleaving with other orders, thereby forming a consecutive block.

To this end, we formalize \emph{online MAPF with order-contiguity (online MAPF-OC)}.
To solve this, we propose a search-based anytime algorithm, called \emph{Dual-Ordering Prioritized Planning (DOPP)}, consisting of a three-level optimization process:
\emph{(i)}~it first searches over order-level priorities to determine the arrival sequence of orders at workstations;
\emph{(ii)}~it refines priorities within orders at the agent level; and
\emph{(iii)}~given these priorities, it synthesizes collision-free and order-contiguous paths via prioritized planning (PP)~\cite{erdman87multiple,silver05cooperative}.
DOPP is complete and runs in polynomial time; it quickly constructs a feasible solution and then iteratively improves the makespan over time.
We evaluate DOPP on various conveyor-network maps, including realistic layouts derived from operational warehouses, demonstrating its practical scalability within tight time budgets and strong makespan performance relative to an ad-hoc reactive policy.

\section{Related Work}
\paragraph{Parcel transport.}
Routing and control for baggage/parcel transport systems have been widely studied, including approaches for online route choice and flow coordination in track-like networks~\cite{tarau10modelbase,zeinaly15integrated}.
These approaches typically focus on coarse routing decisions, while junction-level coordination relies on local rule-based checks~\cite{klotz13automated}, thereby limiting opportunities for globally efficient scheduling.
For conveyor systems, often targeting specific conveyor layouts, several optimization-based schemes have been studied~\cite{ago07simultaneous,novak23optimization}, which reduce routing to structured problems solvable by off-the-shelf solvers, such as MILP or network-flow formulations.
However, this direction introduces a large number of variables and is computationally prohibitive, making it unsuitable for online planning at scale.
This motivates our search-based MAPF approach for online synthesis of collision-free parcel trajectories on general directed graphs.

\paragraph{Multi-agent path finding (MAPF)} is a generic abstraction for coordinating multiple entities on a shared graph under collision avoidance with prominent applications in warehouse automation~\cite{stern19mapf}.
To date, various MAPF solvers have been proposed, ranging from computationally heavy but optimal approaches~\cite{sharon15conflict} to scalable but suboptimal approaches~\cite{okumura2023lacam}.
To capture persistent warehouse-like operations, \emph{lifelong MAPF} has been proposed, where each agent receives a new goal upon reaching its current one~\cite{ma17lifelong}.
\emph{Online MAPF} models a stream of agents that appear over time, while previously revealed agents may already be executing their plans~\cite{svancara19online}.
Our conveyor routing problem inherits this lifelong/online nature but additionally introduces an operational requirement on the arrival sequence at each destination.
As another closely related variant, \emph{MAPF-PC (with precedence constraints)}~\cite{zhang22precedence} assumes that each agent is given a sequence of goals together with precedence constraints between goal-completion events.
In contrast, our focus, MAPF-OC, imposes an order-level non-interleaving constraint at each destination: parcels sharing an order must arrive as a contiguous block, while the order of these blocks is a decision variable optimized for the makespan objective.
Beyond warehouse robotics, MAPF-based techniques have been applied to broad domains, such as rail networks~\cite{li21scalable} and autonomous intersection coordination~\cite{li23intersection,yan24multi}; we similarly extend the application scope of MAPF to conveyor systems.

{
\begin{figure}[t]
  \centering
  \includegraphics{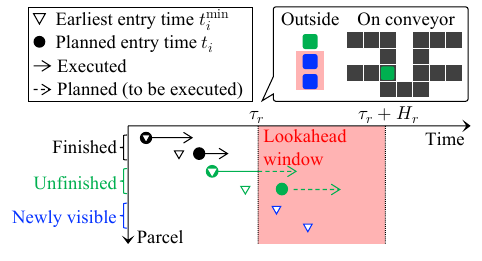}
  \caption{Online replanning at $\tau_r$. Executed trajectories up to $\tau_r$ are fixed and form the initial state of the snapshot. Visible agents are classified as \emph{finished} (already disappeared; black), \emph{unfinished} (still on the conveyor or not entered yet; green), and \emph{newly visible} (revealed within the lookahead $H_r$; blue). The snapshot at $\tau_r$ plans trajectories for unfinished and newly visible agents from $\tau_r$ onward.}
  \label{fig:snapshot_info}
\end{figure}
}
{
\begin{figure*}[t]
  \centering
  \newcommand{\wtab}{0.20\linewidth}
  \newcommand{\wfig}{0.80\linewidth}
  \begin{minipage}[t]{\wtab}
  \vspace{0pt}
    \centering
    \fontsize{8pt}{9pt}\selectfont
    Visible Agents \(\mathcal{A}^r\)\\
    {
    \setlength{\tabcolsep}{4pt}
    \renewcommand{\arraystretch}{0.65}
    \begin{tabular}{@{}ccccc@{}}
      \toprule
      $i$ & $o_i$ & $s_i$ & $g_i$ & $t_i^{\min}$ \\
      \midrule
      $1$ & \includegraphics[width=0.2cm]{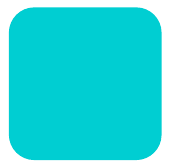} & $u^{\uparrow}$   & $v^{\downarrow}$ & $20$ \\
      $2$ & \includegraphics[width=0.2cm]{figure/orderA.pdf} & $u^{\uparrow}$ & $v^{\downarrow}$ & $28$ \\
      $3$ & \includegraphics[width=0.2cm]{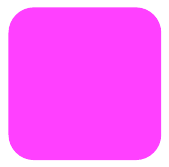} & $u^{\downarrow}$   & $v^{\downarrow}$ & $25$ \\
      $4$ & \includegraphics[width=0.2cm]{figure/orderB.pdf} & $u^{\downarrow}$   & $v^{\downarrow}$ & $29$ \\
      $5$ & \includegraphics[width=0.2cm]{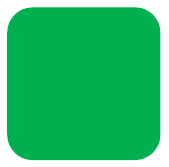} & $u^{\downarrow}$   & $v^{\uparrow}$   & $31$ \\
      $6$ & \includegraphics[width=0.2cm]{figure/orderC.pdf} & $u^{\uparrow}$ & $v^{\uparrow}$   & $34$ \\
      \bottomrule
    \end{tabular}
    }
  \end{minipage}\hfill
  \begin{minipage}[t]{\wfig}
  \vspace{1pt}
  \centering
  \includegraphics[width=\linewidth]{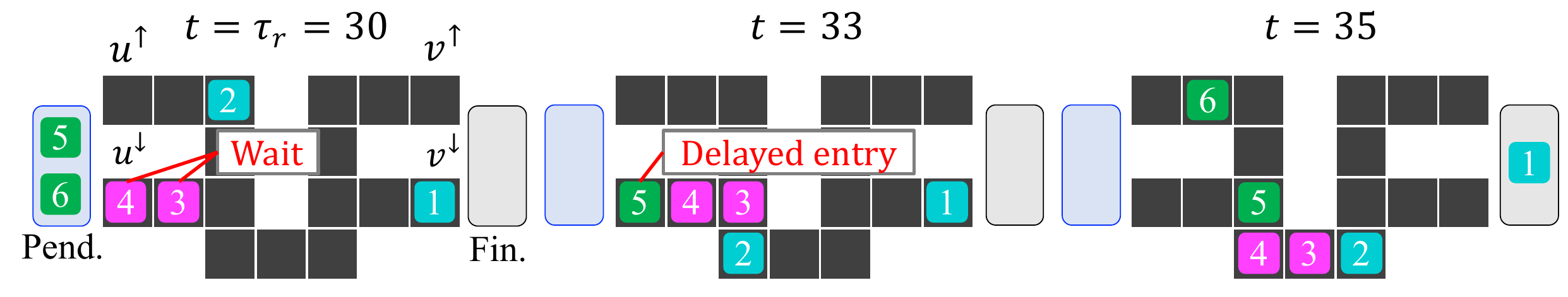}
  \end{minipage}

  \caption{Example of snapshot with $\tau_r=30$, $H_r=5$, and $k=4$. The colors correspond to order labels. Agent $1$ arrives at its destination at $t=30$ and stays in service during $t=31,\ldots,34$, and disappears at $t=35$. To preserve order-contiguity at destination $v^{\downarrow}$, agents $3$ and $4$ wait until agent $2$ passes the junction. This waiting blocks the entry cell, so although $t_{5}^{\min}=31$, the entry of agent $5$ is delayed to $t_{5}=33$. By $t=35$, the remaining parcels are queued so that future arrivals at destinations form consecutive order blocks (i.e., order-contiguous).}
  \label{fig:pd_illustration}
\end{figure*}
}

\section{Problem Definition}\label{sec:problem_definition}
This section formalizes MAPF-OC.
We present the system model and constraints, and define the offline/online problems.

\subsection{System Model}\label{sec:model}
\paragraph{Conveyor network.}
We model the conveyor network as a \emph{directed} graph $G=(V, E)$.
Each vertex $v\in V$ represents a unit-capacity conveyor segment, and each directed edge $(u,v)\in E$ represents one-way transport.
We distinguish \emph{entry} vertices $\mathcal{S}\subset V$, where parcels are injected into the conveyor network, and
\emph{destination} vertices $\mathcal{G}\subset V$, corresponding to workstations in the physical system
(leftmost and rightmost cells in \cref{fig:overview}, respectively).

\paragraph{Agents and orders.}
Each parcel corresponds to an agent $i\in\mathcal{A}$ that moves on $G$ in discrete time $t\in\mathbb{N}_{\ge0}$ and is associated with a tuple $\langle o_i, s_i, g_i, t_i^{\min} \rangle$, where $o_i\in\mathcal{O}$ is an order label, $s_i\in\mathcal{S}$ is an entry vertex, and $g_i\in\mathcal{G}$ is its destination.
We assume a fixed order-to-destination assignment $d:\mathcal{O}\to\mathcal{G}$ and set $g_i=d(o_i)$, and require that $g_i$ is reachable from $s_i$ in $G$.
Each agent has an \emph{earliest entry time} $t_i^{\min}$.
The actual entry time $t_i$ is a decision variable and can be any $t_i\ge t_i^{\min}$.
Upon reaching its destination, an agent must stay there for a \emph{service time} of $k$ steps and then leaves the network.
This models the handling time at the workstation (e.g., unloading/packing).
We assume a uniform $k$ for simplicity.

\subsection{Constraints}\label{sec:problem-formulation}
\paragraph{Time-dependent path} for an agent $i\in\mathcal{A}$ is a function $\pi_i:\mathbb{N}_{\ge0}\to V\cup\{\bot\}$, where $\pi_i(t)$ denotes the location at time $t$ and $\bot$ represents ``outside the network.''
We define the \emph{destination arrival time} $T_i:=\min\{t\mid \pi_i(t)=g_i\}$.
A path $\pi_i$ is \emph{valid} if there exists an entry time $t_i\ge t_i^{\min}$ such that
\begin{align*}
  &\pi_i(t)=\bot \ (t<t_i),\quad \pi_i(t_i)=s_i,\\
  &\pi_i(t+1)\in\{\pi_i(t)\}\cup\{v\mid(\pi_i(t),v)\in E\} (t_i < t < T_i),\\
  &\pi_i(t)=g_i \; (T_i\le t\le T_i+k),\quad \pi_i(t)=\bot \; (t > T_i+k).
\end{align*}
That is, an agent waits or moves to an outgoing neighbor each step, must stay at its destination during service, and then leaves the network.
A set of paths $\pi=\{\pi_i\mid i\in\mathcal{A}\}$ is \emph{collision-free} if
$|\{i\in\mathcal{A}\mid \pi_i(t)=v\}|\le 1$ for all $v \in V$ and $t \in \mathbb{N}_{\ge0}$.
Note that \emph{swap collisions} are impossible on a directed one-way network.

\paragraph{Order-contiguity.}
Given a destination $g\in\mathcal{G}$, let $(i_1,\ldots,i_m)$ be the agents with destination $g$ sorted by increasing arrival times, i.e., $T_{i_1}<\cdots<T_{i_m}$.
We say $\pi$ satisfies \emph{order-contiguity} at $g$ if for any $1\le p<q<r\le m$,
$o_{i_p}=o_{i_r} \Rightarrow o_{i_p}=o_{i_q}$.
Thus, the arrival sequence at each destination can be partitioned into consecutive order blocks.
The set of paths $\pi$ is \emph{order-contiguous} if it satisfies the above for all $g\in\mathcal{G}$.

\subsection{Offline MAPF with Order-Contiguity}\label{sec:oneshot_mapfoc}
The offline setting corresponds to the \emph{one-shot} MAPF problem: an instance specifies $G$ and a finite set of agents $\mathcal{A}$.
A set of paths $\pi=\{\pi_i\mid i\in\mathcal{A}\}$ is a \emph{feasible solution} if it is valid, collision-free, and order-contiguous.
We define the \emph{makespan} of $\pi$ as
$\max_{i\in\mathcal{A}}\,(T_i+k)$.

\subsection{Online MAPF with Order-Contiguity}\label{sec:online_mapfoc}
\paragraph{Batch-based replanning.}
In practice, conveyor systems operate continuously for extended periods (e.g., several hours to a full day).
Planning the entire operation in one-shot is impractical, both because the problem scale becomes enormous and because future order information may be incomplete or change over time.
We therefore replan in \emph{batches} over rounds $r\in\{1,\ldots,r_{\mathrm{end}}\}$.
Let $\tau_r$ be the replanning timestep of round $r$, with $\tau_1:=0$ and $\tau_{r+1}>\tau_r$.
We define the lookahead as $H_r:=\tau_{r+1}-\tau_r$.
At replanning time $\tau_r$, we call agents whose earliest entry time is before $\tau_r+H_r$, \emph{visible}, defined as:
\begin{equation}\label{eq:A-vis}
  \mathcal{A}^r:=\{i\in\mathcal{A}\mid t_i^{\min}<\tau_r+H_r\}.
\end{equation}
We assume that all agents of the same order are revealed within a single replanning window (\emph{single-window reveal}).
This assumption is not merely for modeling simplicity: without sufficient order information, the order-contiguity constraint can become ill-posed online, since newly revealed parcels of an already-processed order may force unavoidable violations of contiguity.
Formally, for each $o\in\mathcal{O}$, there exists an index $r_o$ such that
\begin{equation}\label{eq:single_window}
\tau_{r_o} \le t_i^{\min} < \tau_{r_o+1} \qquad \forall i \in \{ j\in\mathcal{A}\ \mid o_j=o\}.
\end{equation}

\paragraph{Snapshot problem.}
The \emph{initial state} of the snapshot at replanning time $\tau_r$ is given by the executed trajectories until $\tau_r$, as depicted in \cref{fig:snapshot_info}.
In addition, newly visible agents in the current lookahead window are given.
Agents that have already finished service and disappeared by $\tau_r$ remain known, but are excluded from replanning.
At replanning time $\tau_r$, a \emph{snapshot plan} $\pi^r$ is a feasible solution if
\emph{(i)} $\pi_i^r(t)=\pi_i^{r-1}(t)$ for all agents $i\in\mathcal{A}^{r-1}$ and all $t \le \tau_r$, and
\emph{(ii)} the resulting trajectories for agents in $\mathcal{A}^r$ are valid, collision-free, and order-contiguous over the entire time horizon.
Let $T_i^r$ denote the arrival time of agent $i$ under $\pi_i^r$, and define the \emph{snapshot makespan} as
$\max_{i\in\mathcal{A}^r}\,(T_i^r+k)$.
A feasible solution is \emph{snapshot-optimal} if it minimizes this makespan.
\Cref{fig:pd_illustration} illustrates examples of a snapshot instance and a feasible solution.

\paragraph{Lifelong objective.}
In this \emph{online} setting, agents keep becoming visible over replanning steps, and we assume the arrival stream eventually terminates.
We aim to minimize the completion time of the last agent, although no online algorithm can guarantee to achieve this global optimum in general~\cite{svancara19online}.
We therefore optimize the snapshot makespan at each replanning time as a practical surrogate for this objective.

\paragraph{Computational complexity.}
Finding a snapshot-optimal solution is NP-hard.
This follows from the NP-hardness of makespan-minimizing MAPF shown in~\cite{ma16multi}: their reduction can be adapted to our snapshot setting by considering a single replanning window and assigning each agent to a distinct order.
Thus, in this work, we do not aim to compute optimal solutions; rather, our focus is on establishing suboptimal yet scalable real-time planning.

\section{Dual-Ordering Prioritized Planning}
In MAPF-OC, agents can wait outside the network before entering and leave their destination after a bounded service time.
These structural conditions correspond to a \emph{well-formed} setting, under which \emph{prioritized planning (PP)} is guaranteed to find a feasible solution in polynomial time~\cite{cap15prioritized}.
Motivated by this, we develop a PP-based \emph{anytime} algorithm, termed \emph{Dual-Ordering Prioritized Planning (DOPP)}, which quickly constructs a feasible solution to the snapshot problem defined in \cref{sec:online_mapfoc}, then improves the makespan as time permits.
Repeated over replanning rounds, DOPP yields a lifelong execution; the offline setting is the special case of a single round.

We focus on PP for practicality, among the mainstream MAPF paradigms:
CBS and its extensions~\cite{sharon15conflict,li21eecbs} are ill-suited to real-time planning with hundreds of agents, while PIBT and its successors~\cite{okumura22priority,okumura2023lacam} are known to perform poorly in narrow corridors~\cite{okumura23lacam2}.
Besides, enforcing order-contiguity in PIBT necessitates conservative operation, which reduces efficiency, as we observe empirically in \cref{sec:evaluation}.

In the following, we first provide an overview of DOPP, and then describe initialization for obtaining a feasible plan, followed by theoretical properties and anytime refinement.

\subsection{Overview}
\paragraph{Structure.}
DOPP searches priority constraints at high level and runs PP at low level, similar to PBS~\cite{ma19pbs} for classical MAPF.
To enforce order-contiguity, DOPP maintains precedence constraints at both order level and agent level, yielding a three-level structure:
Level~3 searches \emph{order precedence}, Level~2 searches \emph{agent priority} subject to Level~3, and Level~1 applies PP under the resulting priorities.
Pseudocode and conceptual illustration are shown in \cref{alg:dopp} and \cref{fig:algorithm}, respectively.

\paragraph{Priority constraints.}
Let $\mathcal{O}^r\hspace{-1mm}:=\hspace{-1mm}\{o_i \mid i\in\mathcal{A}^r\}$ be the \emph{visible orders} at $\tau_r$, where $\mathcal{A}^r$ is the set of visible agents in Eq.~\eqref{eq:A-vis}.
DOPP maintains two sets of constraints:
an \emph{order-precedence set} $\bm{\prec}_{\mathcal{O}}^r$ over $\mathcal{O}^r$ and an \emph{agent-priority set} $\bm{\prec}_{\mathcal{A}}^r$ over $\mathcal{A}^r$.
An element $o\prec o'$ of $\bm{\prec}_{\mathcal{O}}^r$ (or $i\prec j$ of $\bm{\prec}_{\mathcal{A}}^r$) means that $o$ (resp. $i$) has higher priority than $o'$ (resp. $j$).

\paragraph{Replanning policy.}
In the online setting, while it is possible to replan priorities of all visible agents at each replanning time $\tau_r$, DOPP instead adopts a \emph{committed-suffix} policy to preserve well-formedness.
That is, agents that are already on the network at $\tau_r$ are given higher priority so that they are cleared from the network first.
The relative priorities among these agents are inherited from the previous plan for consistency.
As a result, DOPP optimizes priorities for agents/orders that have not entered the network at $\tau_r$ yet.

{
\begin{figure}[t]
  \centering
  \includegraphics{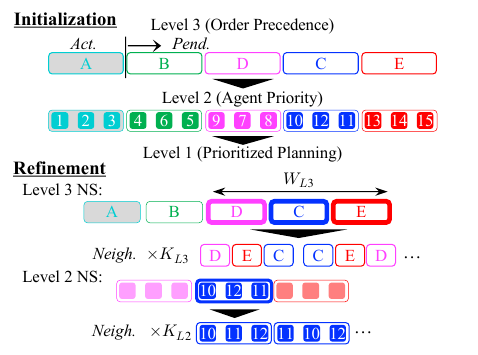}
  \caption{
  Workflow of DOPP.
  Initialization constructs an order precedence (Level~3) and an order-consistent agent priority (Level~2), and then computes a feasible snapshot solution by PP (Level~1).
  Refinement performs neighborhood search: Level~3 perturbs the pending-order subsequence (window permutation), while Level~2 perturbs the agent order within a selected pending order.
  }
  \label{fig:algorithm}
\end{figure}
}
{
\begin{algorithm}[t!]
\caption{High-level of DOPP (snapshot at $\tau_r$)}
\label{alg:dopp}
\small
\begin{algorithmic}[1]
\AlgInput Visible agents $\mathcal{A}^r$, previous snapshot solution $\pi^{r-1}$
    \State Initialize precedence constraints $\bm{\prec}_{\mathcal{O}}^r,\bm{\prec}_{\mathcal{A}}^r$\Comment{\Cref{sec:dopp_init}}\label{algo:dopp:init-precedence}
    \State Topologically sort $\bm{\prec}_{\mathcal{A}}^r$ to obtain $\mathbf{L}_{\mathcal{A}}^r$ (a total order over $\mathcal{A}^r$)
    \State $\pi^r \gets$ run Level~1 with $\mathbf{L}_{\mathcal{A}}^r$ and $\bm{\prec}_{\mathcal{O}}^r$
    \label{algo:dopp:init-level1}

    \While{$\neg\texttt{interrupt()}$}\label{algo:dopp:main-loop}
        \State Construct neighbor constraints $\tilde{\bm{\prec}}_{\mathcal{O}}^r,\tilde{\bm{\prec}}_{\mathcal{A}}^r$ \Comment{\Cref{sec:dopp_refine}}\label{algo:dopp:ns}
        \State Topologically sort $\tilde{\bm{\prec}}_{\mathcal{A}}^r$ to obtain $\tilde{\mathbf{L}}_{\mathcal{A}}^r$
        \State $\tilde{\pi}^r \gets$ run Level~1 with $\tilde{\mathbf{L}}_{\mathcal{A}}^r$ and $\tilde{\bm{\prec}}_{\mathcal{O}}^r$
        \If{$\mathrm{makespan}(\tilde{\pi}^r) < \mathrm{makespan}(\pi^r)$}
            \State $\pi^r \gets \tilde{\pi}^r$;\ \ 
            $\bm{\prec}_{\mathcal{O}}^r \gets \tilde{\bm{\prec}}_{\mathcal{O}}^r$;\ \
            $\bm{\prec}_{\mathcal{A}}^r \gets\tilde{\bm{\prec}}_{\mathcal{A}}^r$\label{algo:dopp:update-solution}
        \EndIf
    \EndWhile
    \State \Return $\pi^r$\Comment{Snapshot solution at $\tau_r$}
\end{algorithmic}
\normalsize
\end{algorithm}
}
\paragraph{Agents and orders classification.}
Herein, to illustrate DOPP, we call an agent \emph{active} if it is already on the conveyor at $\tau_r$; otherwise it is \emph{pending}:
\begin{equation}
\mathcal{A}_{\mathrm{act}}^r \hspace{-1mm}:=\hspace{-1mm} \{ i\in\mathcal{A}^r \hspace{-1mm}\mid\hspace{-1mm} \pi_i^{r-1}(\tau_r)\in V \},
\mathcal{A}_{\mathrm{pend}}^r \hspace{-1mm}:=\hspace{-1mm} \mathcal{A}^r\hspace{-0.5mm}\setminus\hspace{-0.5mm} \mathcal{A}_{\mathrm{act}}^r. \label{eq:class_agents}
\end{equation}
Likewise, an order is \emph{active} if it has at least one active agent; otherwise it is \emph{pending}:
\begin{equation}
\mathcal{O}_{\mathrm{act}}^r \hspace{-1mm}:=\hspace{-1mm} \{ o\in\mathcal{O}^r \hspace{-1mm}\mid\hspace{-1mm} \exists i \in \mathcal{A}_{\mathrm{act}}^r\hspace{-0.5mm}:\hspace{-0.5mm} o_i=o \},
\mathcal{O}_{\mathrm{pend}}^r \hspace{-1mm}:=\hspace{-1mm} \mathcal{O}^r\hspace{-0.5mm}\setminus\hspace{-0.5mm}\mathcal{O}_{\mathrm{act}}^r. \label{eq:class_orders}
\end{equation}

\subsection{Initialization}\label{sec:dopp_init}
We detail how to quickly compute an initial solution via PP (lines~\ref{algo:dopp:init-precedence}--\ref{algo:dopp:init-level1} in \cref{alg:dopp}).
As PP conducts agent-wise path planning based on priorities assigned to each agent, we need to obtain a total ordering over agents, while being careful about maintaining the feasibility of PP.
The Level~2 and Level~3 procedures construct such an order.

\paragraph{Level 3: order precedence.}
We initialize the order-precedence constraints $\bm{\prec}_{\mathcal{O}}^r$ as follows:
\emph{(i)}~place all active orders before pending ones to avoid breaking ongoing arrival blocks, 
\emph{(ii)}~inherit the relative precedence among active orders from the previous round, and
\emph{(iii)}~set precedence among pending orders by an earliest-entry heuristic.
Formally, let $\bm{\prec}_{\mathcal{O}}^{-1}:=\emptyset$ for $r=0$.
For $r\ge 0$, we define $\bm{\prec}_{\mathcal{O}}^r$ as the union of the following constraints:
\begin{align}
\bm{\prec}_{\mathcal{O}}^r
:=
\{& o\prec o' \mid o\in\mathcal{O}_{\mathrm{act}}^r, o'\in\mathcal{O}_{\mathrm{pend}}^r \}\label{eq:precO_activefirst}\\
\cup \{& o\prec o' \mid o,o'\in\mathcal{O}_{\mathrm{act}}^r, (o\prec o')\in \bm{\prec}_{\mathcal{O}}^{r-1} \}\label{eq:precO_active}\\
\cup \{& o\prec o' \mid o,o'\in\mathcal{O}_{\mathrm{pend}}^r, \rho_o < \rho_{o'} \}, \label{eq:precO_pending}
\end{align}
where $\rho_o:=\min\{\,t_i^{\min}\mid i\in\mathcal{A}^r,\ o_i=o\,\}$.
Ties in $\rho_o$ are broken uniformly at random.

\paragraph{Level 2: agent priority.}
We initialize the agent-priority constraints $\bm{\prec}_{\mathcal{A}}^r$ to be consistent with the order precedence and to preserve feasibility across replanning rounds.
Intuitively, we
\emph{(i)}~enforce order-consistency, i.e., agents of higher-precedence orders are planned before those of lower-precedence orders;
then, \emph{(ii)}~for \emph{active} orders $o\in\mathcal{O}_{\mathrm{act}}^r$, those having at least one agent on the conveyor at $\tau_r$, inherit the pairwise precedence among their agents from the previous round; and
\emph{(iii)}~for \emph{pending} orders $o\in\mathcal{O}_{\mathrm{pend}}^r$, those having no agent on the conveyor at $\tau_r$, initialize precedence among their agents with $t_i^{\min}$.
Formally, let $\bm{\prec}_{\mathcal{A}}^{-1}:=\emptyset$ for $r=0$.
For $r\ge 0$, define $\bm{\prec}_{\mathcal{A}}^r$ as the union of the following constraints:
\begin{align}
\bm{\prec}_{\mathcal{A}}^r
:=
\{& i\prec j \mid i,j \in \mathcal{A}^r, (o_i \prec o_j) \in \bm{\prec}_{\mathcal{O}}^r \}\label{eq:precA_inherit}\\
\cup \{& i\prec j \mid o_i=o_j,o_i\in\mathcal{O}_{\mathrm{act}}^r, (i\prec j)\in\bm{\prec}_{\mathcal{A}}^{r-1} \}\label{eq:precA_active}\\
\cup \{& i\prec j \mid o_i=o_j,o_i\in\mathcal{O}_{\mathrm{pend}}^r, t_i^{\min} < t_j^{\min} \}.\label{eq:precA_pending}
\end{align}
We then apply topological sort to $\bm{\prec}_{\mathcal{A}}^r$ in order to obtain a total ordering over $\mathcal{A}^r$, denoted by $\mathbf{L}_{\mathcal{A}}^r$, which is used in Level~1.
Ties in $t_i^{\min}$ are broken uniformly at random.

\paragraph{Level 1: PP.}\label{sec:dopp_level1}
Level~1 synthesizes a snapshot solution by applying standard PP under $\mathbf{L}_{\mathcal{A}}^r$:
agents are planned sequentially using space-time A*~\cite{silver05cooperative} while avoiding previously planned paths.
We must additionally enforce order-contiguity by applying \emph{destination blocking},
which forbids agent $i$ from occupying its destination $g_i$ before time $\Gamma_i$:
\[
\Gamma_i := \max\{T_j^r + k + 1 \mid j\in\mathcal{A}^r, g_j=g_i, (o_j \prec o_i) \in \bm{\prec}_{\mathcal{O}}^r \}.
\]
If the set is empty, we set $\Gamma_i:=0$.
Since $\mathbf{L}_{\mathcal{A}}^r$ respects $\bm{\prec}_{\mathcal{O}}^r$, this yields order-contiguous arrivals.

\subsection{Theoretical Analysis}\label{sec:theoretical_analysis}
DOPP's initialization in \cref{sec:dopp_init} ensures that the snapshot instance at each $\tau_r$ falls into a well-formed setting: agents may wait outside the network, and destinations are released after bounded service.
This reduces the snapshot feasibility of MAPF-OC to feasibility under PP (Level~1).
Therefore, running Level~1 once under the initialized priorities yields a feasible snapshot solution in polynomial time.
Full proofs are provided in Appendix.

\begin{lemma}\label{lem:precO_total}
For every round $r$, $\bm{\prec}_{\mathcal{O}}^r$ and $\bm{\prec}_{\mathcal{A}}^r$ admit a total ordering over $\mathcal{O}^r$ and $\mathcal{A}^r$, respectively, via topological sorting.
\end{lemma}

\begin{lemma}\label{lem:blocking_oc}
At every round $r$, the snapshot solution produced by Level~1 is order-contiguous.
\end{lemma}

\begin{theorem}\label{thm:feasible_poly}
At each replanning time $\tau_r$, DOPP constructs a feasible snapshot solution in polynomial time.
\end{theorem}

\subsection{Refinement}\label{sec:dopp_refine}
The solution obtained by initialization is feasible but is based on heuristically chosen priorities, leaving room for makespan improvement.
In the batch-based replanning assumed in \cref{sec:online_mapfoc}, planning for the next batch can proceed while the current batch is being executed, providing additional computation time.
DOPP therefore performs an \emph{anytime refinement} (lines~\ref{algo:dopp:main-loop}--\ref{algo:dopp:update-solution} in \cref{alg:dopp}) that searches better priority constraints at Level~3 and/or Level~2, while preserving well-formedness.
Implementation details are provided in Appendix.

\paragraph{Order-precedence refinement.}
Level~3 refines $\bm{\prec}_{\mathcal{O}}^r$ by modifying the constraints in
\cref{eq:precO_pending},
while keeping the constraints in \cref{eq:precO_activefirst,eq:precO_active} unchanged; this preserves well-formedness by retaining precedence among active orders.
We perform \emph{neighborhood search (NS)} over pending orders by selecting a contiguous window of $W_{\mathrm{L3}}$ orders and sampling $K_{\mathrm{L3}}$ candidate perturbations within it (\cref{fig:algorithm}, bottom).
For each candidate, we rebuild $\tilde{\bm{\prec}}_{\mathcal{O}}^r$ accordingly, re-initialize the corresponding agent-priority constraints, and evaluate it by running Level~1.
These candidate evaluations are independent and can be parallelized.

{
\newcommand{\figheight}{0.1\linewidth}
\newcommand{\figwidth}{0.157\linewidth}
\newcommand{\colwidth}{0.159\linewidth}
\newcommand{\xlabel}{-1.1 * \colwidth}
\newcommand{\xscen}{-1.25 * \colwidth}
\newcommand{\entry}[9]{%
  \node[anchor=north] at (\colwidth * #1 - 0.49 * \colwidth, \figheight * 1.2) {\scriptsize\emph{#3}};
  \node[anchor=north] at (\colwidth * #1 - 0.49 * \colwidth, \figheight * 1.05)
  {\tiny $|V|$$=$$#4$, $\lambda$$=$$#5$, $|\mathcal{A}^r|$$=$$#6$};
  %
  \node[anchor=north east] at (\colwidth * #1, \figheight * 0.9)
  {\includegraphics[width=0.13\linewidth]{result_one-shot/map/#7}};
  %
  \node[anchor=north east] at (\colwidth * #1, -\figheight * 0)
  {\includegraphics[width=\figwidth]{result_one-shot/gap/#8}};
  \node[anchor=north east] at (\colwidth * #1, -\figheight * 1)
  {\includegraphics[width=\figwidth]{result_one-shot/improvement/#8}};
  \node[anchor=north east] at (\colwidth * #1, -\figheight * 2)
  {\includegraphics[width=\figwidth]{result_one-shot/utilization/#8}};
  %
  \node[anchor=north east] at (\colwidth * #1, -\figheight * 3.2)
  {\includegraphics[width=\figwidth]{result_lifelong/gap/#9}};
  \node[anchor=north east] at (\colwidth * #1, -\figheight * 4.2)
  {\includegraphics[width=\figwidth]{result_lifelong/utilization/#9}};
}
    \begin{figure*}[th!]
    \centering
    \begin{tikzpicture}
    \entry{0}{0}{Small-A}{$69$}{1}{$300$}{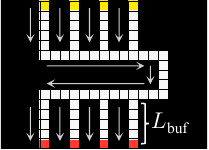}{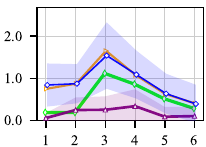}{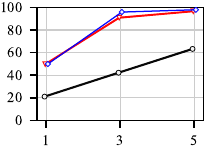} 
    \entry{1}{0}{Small-B}{$72$}{1}{$300$}{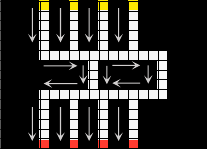}{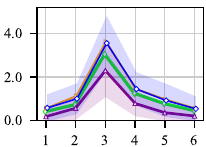}{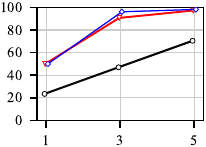}
    \entry{2}{0}{Medium-A}{$159$}{1}{$300$}{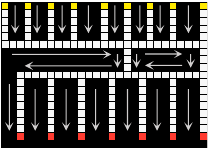}{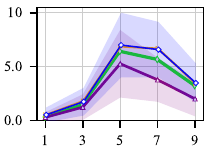}{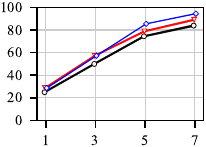}
    \entry{3}{0}{Medium-B}{$162$}{1}{$300$}{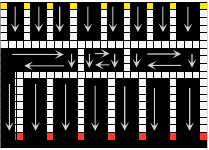}{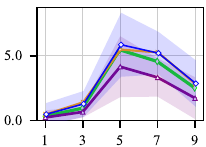}{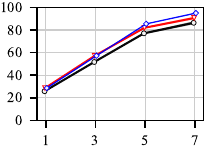}
    \entry{4}{0}{Complex}{$229$}{3}{$500$}{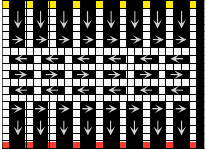}{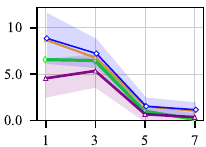}{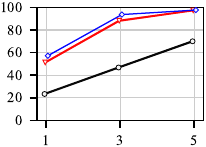}
    \entry{5}{0}{Large}{$360$}{3}{$1000$}{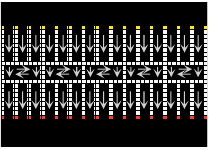}{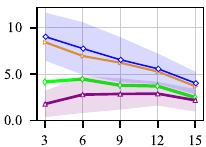}{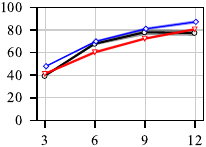}
    %
    \node[rotate=90,anchor=west] at (\xscen, -2.0 * \figheight) {One-shot};
    \node[rotate=90,anchor=west] at (\xscen, -4.6 * \figheight) {Lifelong};
    \node[rotate=90,anchor=west] at (\xlabel, -1.1 * \figheight) {\scriptsize $\leftarrow$Makespan/LB};
    \node[rotate=90,anchor=west] at (\xlabel, -2.0 * \figheight) {\scriptsize Impr.(\%)$\rightarrow$};
    \node[rotate=90,anchor=west] at (\xlabel, -3.1 * \figheight) {\scriptsize Avg. Util.(\%)$\rightarrow$};
    \node[rotate=90,anchor=west] at (\xlabel, -4.3 * \figheight) {\scriptsize $\leftarrow$Makespan/LB};
    \node[rotate=90,anchor=west] at (\xlabel, -5.3 * \figheight) {\scriptsize Avg. Util.(\%)$\rightarrow$};
    \node[] at (\colwidth * -0.45, -5.6 * \figheight) {\small Service time $k$};
    \node[anchor=east] at (\colwidth * 5, -5.62 * \figheight)
    {\includegraphics[width=0.78\linewidth]{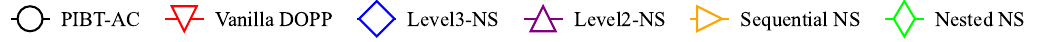}};
    \end{tikzpicture}
    \caption{Results for one-shot and lifelong experiments. 
    Headers give the map name, graph size $|V|$, arrival rate $\lambda$, and the number of visible agents per window $|\mathcal{A}^r|$, along with the layout (yellow: entries $\mathcal{S}$; red: destinations $\mathcal{G}$).
    \emph{Upper} rows shows the results for the one-shot problem: makespan $M/M_{\mathrm{LB}}$, relative makespan improvement over \emph{Vanilla DOPP} by anytime refinement, $\mbox{Impr.\,(\%)}:=100\cdot (M_{\mathrm{Vanilla}}-M)/M_{\mathrm{Vanilla}}$, and the utilization $\bar U$. 
    \emph{Lower} rows show those for the lifelong problem: makespan and utilization over 10 consecutive online windows.
    Each curve shows the mean over 20 instances.
    Shaded areas indicate the standard deviation for representative methods.
    }
    \label{fig:result_main}
  \end{figure*}
}

\paragraph{Agent-priority refinement.}
Level~2 refines $\bm{\prec}_{\mathcal{A}}^r$ by modifying the constraints in \cref{eq:precA_pending},
while keeping the constraints in \cref{eq:precA_inherit,eq:precA_active} unchanged; this preserves well-formedness by retaining priorities among active agents.
We perform NS by selecting a pending order $o\in\mathcal{O}_{\mathrm{pend}}^r$ and sampling $K_{\mathrm{L2}}$ alternative orderings of its agents.
For each candidate, we rebuild the constraints $\tilde{\bm{\prec}}_{\mathcal{A}}^r$ accordingly and evaluate it by rerunning Level~1.

\paragraph{Combination strategy.}
Refinement can be applied at Level~3 only, at Level~2 only, or by combining both.
For example, we may \emph{(i)} refine $\bm{\prec}_{\mathcal{O}}^r$ and, for each candidate, rebuild $\bm{\prec}_{\mathcal{A}}^r$ using the initialization rule and evaluate it by rerunning Level~1;
and/or \emph{(ii)} additionally refine $\bm{\prec}_{\mathcal{A}}^r$ for each Level~3 candidate, generating multiple candidates before rerunning Level~1.
We compare these strategies in \cref{sec:evaluation}.

\section{Evaluation}\label{sec:evaluation}
This section evaluates the effectiveness of \emph{DOPP} in terms of snapshot quality, runtime, scalability, anytime behavior, and lifelong performance.
All solvers are implemented in C++ and run on a MacBook Pro (Apple M3 Max, \SI{128}{\giga\byte} RAM).

\subsection{Setup}\label{sec:setup}
\paragraph{Maps.}
We use directed conveyor maps shown in \cref{fig:result_main} (top row), designed to cover both typical warehouse layouts and stress-test settings.
\texttt{Small-A/B} are simple maps that capture basic topology, and \texttt{Medium-A/B} are practical-scale maps derived from existing conveyor systems.
To stress-test beyond typical layouts, we also include \texttt{Complex} and \texttt{Large} maps to evaluate planners in more complex and large-scale environments.
Herein, we denote the length of exit-buffer as $L_{\mathrm{buf}}$ (e.g., $8$ for \texttt{Medium-A/B}).

\paragraph{Parcel generation.}
At each time $t=0,1,\ldots$, we generate $\lambda$ new agents with the earliest entry time $t_i^{\min}=t$.
To maintain a comparable workload density across map scales, $\lambda$ is set to $1$ for \texttt{Small}/\texttt{Medium}\texttt{-A/B}, and $3$ for \texttt{Complex} and \texttt{Large}.
For each planning time $\tau_r$, we set the number of agents visible per window $|\mathcal{A}^r|$ to $300$ for \texttt{Small}/\texttt{Medium}\texttt{-A/B},  $500$ for \texttt{Complex}, and $1{,}000$ for \texttt{Large}.
In the lifelong scenario, the next planning time $\tau_{r+1}$ is set to $\max_{i\in\mathcal{A}^r} t_i^{\min}$.
Agents are generated in order blocks: we sample an order size $l_o$ uniformly at random from $\{1,\ldots,L_{\mathrm{buf}}\}$, generate $l_o$ agents for the order, and then switch to the next order.
Each agent chooses its entry $s_i$ at random from the map-specific set $\mathcal{S}$.
Destinations are assigned per order to maintain a balanced workload across $\mathcal{G}$ by selecting the destination with the current minimum cumulative count and updating the count after assignment.
These settings are intended to mimic typical fulfillment operations.

\paragraph{DOPP variants.}
For snapshot experiments, we evaluate five variants of \emph{DOPP}:
\emph{Vanilla} (no refinement, i.e., initial solution),
\emph{Level3-NS} (refine only Level~3; rebuild Level~2 as in initialization for each Level~3 candidate),
\emph{Level2-NS} (refine only Level~2; keep Level~3 fixed),
\emph{Sequential~NS} (first refine Level~3, then refine Level~2 for the best Level~3 candidate), and
\emph{Nested~NS} (for each Level~3 candidate, refine Level~2 to generate multiple candidates).
For neighborhood search, we use $W_{\mathrm{L3}}=7$ and $K_{\mathrm{L3}}=K_{\mathrm{L2}}=10$.
These hyperparameters were set via a pilot study; see Appendix.
We use a \SI{60}{\second} refinement budget per snapshot instance:
\emph{Level2/3-NS} run for \SI{60}{\second}, \emph{Sequential~NS} runs Level~3 for \SI{40}{\second} then Level~2 for \SI{20}{\second}, and \emph{Nested~NS} uses \SI{60}{\second} total for Level~3 with a \SI{1}{\second} cap for each inner Level~2 call. 
Level~3 refinement is parallelized with up to 16 threads.

\paragraph{Baseline.}
We also evaluate the baseline \emph{PIBT-AC}, which runs PIBT~\cite{okumura22priority}---a representative MAPF algorithm that prevents collisions locally---and enforces order-contiguity via \emph{admission control (AC)}.
For each destination $g$, while an order to $g$ is active, agents of other orders destined for $g$ are not admitted into the network, and the lock is released when the remaining agents of the active order are within $L_{\mathrm{buf}}$ steps of $g$.
This provides an intuitive and scalable rule-based solution to MAPF-OC.
Note that adapting other MAPF solvers, such as CBS variants, is non-trivial due to the order-contiguity constraint. 
Furthermore, these methods are not well suited for real-time planning on large-scale problems, which led to the current experimental design.

\paragraph{Metrics.}
We use the \emph{normalized makespan} $M/M_{\mathrm{LB}}$, where $M=\max_i (T_i+k)$, $M_{\mathrm{LB}}=\max_i\bigl(t_i^{\min}+\mathrm{dist}(s_i,g_i)+k\bigr)$, and $\mathrm{dist}$ is the shortest travel time on $G$; hence $M_{\mathrm{LB}}$ serves as the makespan lower bound.
We also report \emph{average destination utilization} $\bar U$~(\%) to diagnose whether performance is limited by destination capacity.
When $\bar U$ is close to 100\%, improving makespan requires increasing destination capacity (e.g., adding workstations) rather than better routing.
Concretely, for each destination $g$, let $x_g(t)=1$ iff $\exists i:\pi_i(t)=g$ (otherwise $0$), and denote
$t_a(g)=\min\{t\mid x_g(t)=1\}$ and $t_b(g)=\max\{t\mid x_g(t)=1\}$.
The utilization is $U(g):=100\cdot\bigl(\sum_{t=t_a(g)}^{t_b(g)} x_g(t)\bigr)/(t_b(g)-t_a(g)+1)$, and $\bar U$ is the average over destinations.

\subsection{Results}\label{sec:results}
We first solve the snapshot offline problem (\cref{sec:oneshot_mapfoc}) to analyze the behavior of \emph{DOPP variants}, and then evaluate the representative ones on the lifelong setting (\cref{sec:online_mapfoc}).

\paragraph{Snapshot quality.}
\Cref{fig:result_main} (upper) summarizes makespan, improvement over \emph{Vanilla DOPP}, and destination utilization across six maps while sweeping the service time $k$.
Across all maps and $k$, \emph{DOPP} consistently outperforms \emph{PIBT-AC}, and even \emph{Vanilla} achieves strong makespans.
The advantage over \emph{PIBT-AC} is striking on \texttt{Small-A/B}, where admission control induces long waits relative to the map size, whereas \emph{DOPP} remains close to the lower bound for $k \le 3$.
As $k$ increases, the utilization approaches saturation, and the makespan degrades for all methods, suggesting that performance becomes dominated by workstation capacity rather than routing decisions.
The benefit of refinement is most visible at intermediate $k$ values where both routing conflicts and destination bottlenecks matter; in these cases, \emph{Level3-NS} yields the best results, followed by \emph{Sequential NS}.
This highlights that order precedence (Level~3) is a primary driver of performance.
The same trend appears on \texttt{Medium-A/B}, where the gap among \emph{DOPP} variants is particularly pronounced, and it persists on the \texttt{Complex} and \texttt{Large} maps, demonstrating robustness to complex layouts and larger instances.

{
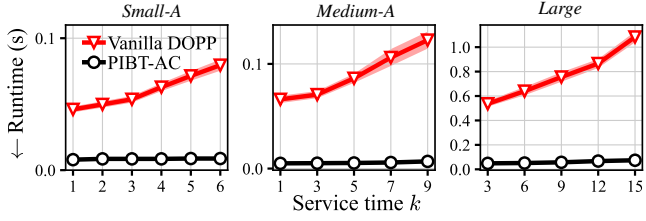
\begin{figure}[t]
  \centering
  \begin{tikzpicture}
    \node[anchor=south] at ($ (2.0,0) + (0, {-5pt}) $) {\scriptsize\emph{Small-A}};
    \node[anchor=south] at ($ (4.7,0) + (0, {-5pt}) $) {\scriptsize\emph{Medium-A}};
    \node[anchor=south] at ($ (7.4,0) + (0, {-5pt}) $) {\scriptsize\emph{Large}};
    \node[anchor=north] at ($ (0,0) + ({0.50\columnwidth}, {0pt}) $) {%
      \includegraphics[width=\columnwidth]{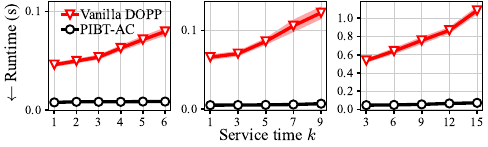}%
    };
  \end{tikzpicture}
  \caption{Runtime at representative snapshot settings in \cref{fig:result_main}. Each point shows the mean over 20 instances.
  }
  \label{fig:oneshot_runtime}
\end{figure}
}

{
\begin{figure}[t]
  \centering
  \includegraphics[width=\linewidth]{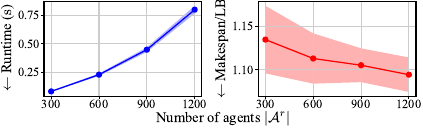}
  \caption{DOPP's scalability on \texttt{Medium-A}. 
  Runtime and makespan against the number of visible agents $|\mathcal{A}^r|$ are shown.
  Service time $k$ is set to five. 
  Each point shows the mean over 20 instances.
  }
  \label{fig:scalability}
\end{figure}
}
{
    \setlength{\tabcolsep}{0pt}
    \newcommand{\entry}[4]{%
    \begin{minipage}{0.49\linewidth}
        \centering
        {\scriptsize\hspace{10mm}\textit{#1}\quad(#3)}
        {\includegraphics[width=1.0\linewidth]{result_anytime/#4_improvement_vs_k.pdf}}
        {\includegraphics[width=1.0\linewidth]{result_anytime/#4_l3-only_anytime.pdf}}
    \end{minipage}
    }%
    \begin{figure}[t]
    \centering
    \begin{tabular}{cc}
      \entry{Medium-A}{$|V|=159$}{$|\mathcal{A}^r|{=}300$}{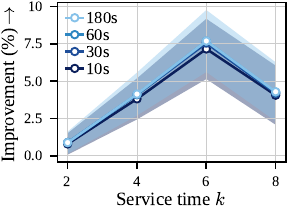} &
      \entry{Large}{$|V|=360$}{$|\mathcal{A}^r|{=}1,000$}{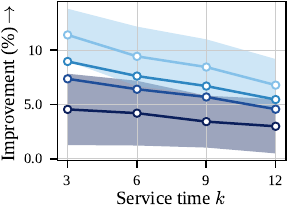}\\ [-0pt]
    \end{tabular}
    \caption{
    Anytime refinement on \texttt{Medium-A} (left) and \texttt{Large} (right), over 20 instances.
    \emph{Top}: relative makespan improvement of \emph{Level3-NS} over \emph{Vanilla DOPP} after a planning budget $10/30/60/180$ \si{\second}.
    \emph{Bottom}: best-so-far makespan versus runtime under the \SI{180}{\second} budget across the various service time $k$.
    }
    \label{fig:result_anytime}
    \end{figure}
}

\paragraph{Runtime and scalability.}
\Cref{fig:oneshot_runtime} reports runtime at representative snapshot cases sampled from the one-shot results in \cref{fig:result_main}.
\emph{PIBT-AC} is extremely lightweight and runs in $\SI{50}{\milli\second}$ even on \texttt{Large} ($|\mathcal{A}^r|=1{,}000$) with $k{=}3$, while \emph{Vanilla DOPP} remains practical (${\approx}\SI{0.53}{\second}$) despite achieving much better solution quality.
Runtime increases mildly as $k$ grows due to longer makespans and deeper space-time searches, but remains below \SI{0.1}{\second} on \texttt{Small/Medium-A} and below \SI{1}{\second} even on \texttt{Large}.
\Cref{fig:scalability} stresses \emph{Vanilla} on \texttt{Medium-A} by increasing the number of visible agents.
Runtime grows with $|\mathcal{A}^r|$ but stays below \SI{1}{\second} even with $|\mathcal{A}^r|=1{,}200$.
This scalability performance covers practically relevant horizons for conveyor control, since far-future parcels are less certain in online operation.
The normalized makespan improves with larger $|\mathcal{A}^r|$, because the denominator $M_{\mathrm{LB}}$ increases with later release times while $M$ stays close to $M_{\mathrm{LB}}$.

\paragraph{Anytime behavior.}
We analyze how solution quality improves with additional computation time.
In \cref{fig:result_anytime}, the top row quantifies the benefit of refinement by plotting the relative makespan improvement of \emph{Level3-NS} over \emph{Vanilla DOPP} under different time budgets.
\texttt{Medium-A} largely saturates by \SI{10}{\second} (${\approx} 7.5\%$ improvement at $k{=}6$), whereas \texttt{Large} continues to improve as time permits (e.g., exceeding $10\%$ improvement at $k{=}3$) due to its instance scale.
The bottom row plots best-so-far makespan versus time, showing that most gains are obtained early on \texttt{Medium-A}.
On \texttt{Large}, rapid initial improvements are followed by slower but persistent progress up to the time limit.

\paragraph{Lifelong performance.}
Finally, we evaluate the online setting by solving 10 consecutive windows\footnote{Pilot experiments with longer scenarios showed similar qualitative trends.}.
\Cref{fig:result_main} (lower) shows that the snapshot improvements translate to lifelong replanning---\emph{DOPP} continues to produce effective replans, resulting in consistently lower makespan than \emph{PIBT-AC}.
Refinement via \emph{Level3-NS} remains beneficial in the lifelong setting; for example, on \texttt{Medium-A} with $k{=}5$, \emph{Vanilla} finds a feasible plan in \SI{0.3}{\second} per window on average, and \emph{Level3-NS} removes about $92\%$ of the excess over the lower bound (from $1.093$ to $1.007$).
\emph{Vanilla} is slightly slower than in one-shot runs (\cref{fig:oneshot_runtime}) because agents carry over across windows, but the overhead is modest since PP dominates its runtime and scales linearly with the number of planned agents.
As an implementation note, replanning intervals $\tau_{r+1}-\tau_r$ span hundreds of timesteps across maps in our setup, so the \SI{60}{\second} refinement budget fits within each replanning cycle.
These results indicate that the gains observed in snapshot planning carry over to settings where decisions must be made repeatedly online, with the system state evolving from previous decisions.

\section{Conclusion}
We formulated online MAPF-OC on conveyor networks and presented \emph{DOPP}, an anytime algorithm that quickly finds feasible solutions and improves makespan as time permits.
Experiments demonstrated practical scalability and strong solution quality.
This work enriches an application example of foundational MAPF algorithms beyond typical warehouse sortation.
Future work includes extending online MAPF-OC to incorporate the parcel entry positions and destination assignment for more advanced logistics automation.

\bibliographystyle{named}
\bibliography{ref_macro,ijcai26}

{
\appendix
\section*{Appendix}
\section{Proofs}
We provide omitted proofs for lemmas and theorems in \cref{sec:theoretical_analysis}.

\paragraph{Lemma~\ref{lem:precO_total}.}
\textit{
For every round $r$, $\bm{\prec}_{\mathcal{O}}^r$ and $\bm{\prec}_{\mathcal{A}}^r$ admit a total ordering over $\mathcal{O}^r$ and $\mathcal{A}^r$, respectively, via topological sorting.
}

\begin{proof}
We prove the claim for $\bm{\prec}_{\mathcal{O}}^r$ and $\bm{\prec}_{\mathcal{A}}^r$ separately.

\noindent\emph{Order precedence $\bm{\prec}_{\mathcal{O}}^r$.}
We argue by induction on $r$.
At $r=0$, visible orders are totally ordered by the key $\rho_o$ with tie-breaking.
Assume $\bm{\prec}_{\mathcal{O}}^{r-1}$ induces a total ordering on the visible orders of round $r-1$.
In round $r$, DOPP places every active order before every pending order (Eq.~\eqref{eq:precO_activefirst}),
inherits the relative precedence among active orders from $\bm{\prec}_{\mathcal{O}}^{r-1}$ (Eq.~\eqref{eq:precO_active}),
and totally orders pending orders by sorting $\rho_o$ with tie-breaking (Eq.~\eqref{eq:precO_pending}).
Thus any two distinct orders in $\mathcal{O}^r$ are comparable, and $\bm{\prec}_{\mathcal{O}}^r$ admits a total ordering via topological sorting.

\noindent\emph{Agent priority $\bm{\prec}_{\mathcal{A}}^r$.}
We argue by induction on $r$.
We show that $\bm{\prec}_{\mathcal{A}}^r$ admits a total ordering over $\mathcal{A}^r$ by establishing that any two distinct agents $i,j\in\mathcal{A}^r$ are comparable.
If $o_i\neq o_j$, then $\bm{\prec}_{\mathcal{O}}^r$ totally orders $o_i$ and $o_j$, and Eq.~\eqref{eq:precA_inherit} transfers this order to $i$ and $j$; hence $i$ and $j$ are comparable in $\bm{\prec}_{\mathcal{A}}^r$.
Now consider the case $o_i=o_j=o$.
If $o\in\mathcal{O}_{\mathrm{act}}^r$, then by the \emph{single-window reveal} assumption (Eq.~\eqref{eq:single_window}), all agents of order $o$ were already visible in round $r-1$ and hence were included in $\bm{\prec}_{\mathcal{A}}^{r-1}$.
By the induction hypothesis, $\bm{\prec}_{\mathcal{A}}^{r-1}$ admits a total ordering over these agents, and Eq.~\eqref{eq:precA_active} preserves their pairwise precedence in round $r$; hence $i$ and $j$ are comparable.
If $o\in\mathcal{O}_{\mathrm{pend}}^r$, then all agents of $o$ are pending and are ordered by $t^{\min}$ with tie-breaking (Eq.~\eqref{eq:precA_pending}), hence $i$ and $j$ are comparable.
Hence every pair of distinct agents in $\mathcal{A}^r$ is comparable, and $\bm{\prec}_{\mathcal{A}}^r$ admits a total ordering via topological sorting.
\end{proof}

\paragraph{Lemma~\ref{lem:blocking_oc}.}
\textit{
At every round $r$, the snapshot solution produced by Level~1 is order-contiguous.
}
\begin{proof}
We show that destination blocking in Level~1, introduced in \cref{sec:dopp_init}, enforces order-contiguity.
Fix a destination $g\in\mathcal{G}$ and consider the agents destined to $g$ sorted by increasing arrival times $T_{i_1}^r < \cdots < T_{i_m}^r$.
Suppose for contradiction that the arrival sequence at $g$ is not order-contiguous.
Then there exist indices $1\le p<q<r\le m$ such that $o_{i_p}=o_{i_r}$ and $o_{i_p}\neq o_{i_q}$.

Let $o:=o_{i_r}(=o_{i_p})$ and $o':=o_{i_q}$.
By Lemma~\ref{lem:precO_total}, either $o\prec o'$ or $o'\prec o$ holds.
If $o\prec o'$, then for agent $i_q$ we have
\[
\Gamma_{i_q}\ \ge\ T_{i_r}^r + k + 1,
\]
because $i_r$ is an agent of the higher-precedence order $o$ assigned to the same destination $g$.
Hence $T_{i_q}^r \ge \Gamma_{i_q} > T_{i_r}^r$, contradicting $T_{i_q}^r < T_{i_r}^r$.
The case $o'\prec o$ symmetrically contradicts $T_{i_p}^r < T_{i_q}^r$.
Therefore no such triple exists and arrivals at $g$ are order-contiguous.
\end{proof}

\paragraph{Theorem~\ref{thm:feasible_poly}.}
\textit{
At each replanning time $\tau_r$, DOPP constructs a feasible snapshot solution in polynomial time.
}

\begin{proof}
We first establish \emph{completeness} (existence of a feasible snapshot solution returned by PP under the constructed priorities), and then bound the computational complexity.

\paragraph{Completeness.}
We prove by induction on $r$ that Level~1 always plans all agents in $\mathcal{A}^r$ without failure, producing a snapshot-feasible solution.
For the case $r=0$, Lemma~\ref{lem:precO_total} ensures that the initialized precedences
$\bm{\prec}_{\mathcal{O}}^0$ and $\bm{\prec}_{\mathcal{A}}^0$ admit total orderings over $\mathcal{O}^0$ and $\mathcal{A}^0$.
Thus PP under the induced total agent order produces a collision-free snapshot plan,
and order-contiguity follows from Lemma~\ref{lem:blocking_oc}.
Hence $\pi^{0}$ is feasible.

\emph{Active agents.}
Consider the active agents $\mathcal{A}_{\mathrm{act}}^r$.
By definition in Eq.~\eqref{eq:class_agents}, every agent in $\mathcal{A}_{\mathrm{act}}^r$ was already planned at $\tau_{r-1}$.
Level~1 plans all active agents before any pending agent, and preserves their relative priority from round $r-1$ (Eq.~\eqref{eq:precO_activefirst} and Eq.~\eqref{eq:precA_active}).
By the induction hypothesis, $\pi^{r-1}$ is snapshot-feasible, so a feasible continuation exists at $\tau_r$ for all active agents.
Therefore, PP succeeds for all agents in $\mathcal{A}_{\mathrm{act}}^r$.

\emph{Pending agents.}
After planning agents in $\mathcal{A}_{\mathrm{act}}^r$, let $t^\star$ be any time strictly after all already planned agents have disappeared from the network.
A finite $t^\star$ exists because every agent that reaches its destination leaves the network after $k$ steps.
When planning a pending agent $i\in\mathcal{A}_{\mathrm{pend}}^r$, choose an entry time
\[
t_i \ge \max\{t_i^{\min}, t^\star, \Gamma_i\}.
\]
Then follow any directed path from $s_i$ to $g_i$, which exists under the reachability assumption in \cref{sec:problem_definition}.
Since the network is empty after $t^\star$, this yields a valid collision-free path that also respects destination blocking.
Repeating this for all pending agents shows that PP succeeds for all agents in $\mathcal{A}_{\mathrm{pend}}^r$.

Finally, order-contiguity follows from Lemma~\ref{lem:blocking_oc}.
Thus the resulting snapshot solution is feasible.

\paragraph{Polynomial time.}
Levels~3 and~2 consist of sorting and topological sorting over $\mathcal{O}^r$ and $\mathcal{A}^r$, thus run in polynomial time.
For Level~1, we provide a finite horizon that always contains a feasible solution.
We define,
\[
C_{\mathrm{act}}^r :=
\begin{cases}
\max_{i\in\mathcal{A}_{\mathrm{act}}^r}\bigl(T_i^{r-1}+k+1\bigr) & \text{if }\mathcal{A}_{\mathrm{act}}^r\neq\emptyset,\\
\tau_r & \text{otherwise},
\end{cases}
\]
and
\[
t_{\max}^r :=
\begin{cases}
\max_{i\in\mathcal{A}_{\mathrm{pend}}^r} t_i^{\min} & \text{if }\mathcal{A}_{\mathrm{pend}}^r\neq\emptyset,\\
\tau_r & \text{otherwise}.
\end{cases}
\]
Let $D:=|V|-1$ and define
\[
\overline{T}^r := \max\{C_{\mathrm{act}}^r,\, t_{\max}^r,\, \tau_r\} + |\mathcal{A}_{\mathrm{pend}}^r|\,(D+k+1).
\]
When planning an agent, Level~1 searches on the time-expanded graph truncated at $\overline{T}^r$,
which has $O(|V|\cdot\overline{T}^r)$ vertices and $O(|E|\cdot\overline{T}^r)$ edges.
Space-time A* search on this graph runs in polynomial time in $|V|,|E|,$ and $\overline{T}^r$.
Repeating this for all agents in $\mathcal{A}^r$ yields an overall polynomial running time.

\end{proof}

\section{Implementation Details of Refinement}
We describe the implementation of the neighborhood searches in Level~3 and Level~2 introduced in \cref{sec:dopp_refine}.

\paragraph{Order-precedence refinement (Level~3).}
Let $\mathbf{L}_{\mathcal{O}}^r$ be a total ordering over $\mathcal{O}^r$ obtained by topologically sorting $\bm{\prec}_{\mathcal{O}}^r$, and let $\mathbf{L}_{\mathcal{O},\mathrm{pend}}^r$ be the subsequence of $\mathbf{L}_{\mathcal{O}}^r$ containing only pending orders in $\mathcal{O}_{\mathrm{pend}}^r$.
We select a contiguous window of $W_{\mathrm{L3}}$ orders in $\mathbf{L}_{\mathcal{O},\mathrm{pend}}^r$,
and sample $K_{\mathrm{L3}}$ perturbed sequences by randomly permuting its elements.
For each perturbed sequence $\tilde{\mathbf{L}}_{\mathcal{O},\mathrm{pend}}^r=(o_1,\ldots,o_m)$,
we construct candidate precedence constraints among pending orders as the chain
$\tilde{\bm{\prec}}_{\mathcal{O},\mathrm{pend}}^r:=\{o_l \prec o_{l+1}\mid l=1,\ldots,m-1\}$,
and form $\tilde{\bm{\prec}}_{\mathcal{O}}^r$ by replacing \cref{eq:precO_pending} with $\tilde{\bm{\prec}}_{\mathcal{O},\mathrm{pend}}^r$
while keeping \cref{eq:precO_activefirst,eq:precO_active} unchanged.
We then re-initialize the agent-priority constraints under $\tilde{\bm{\prec}}_{\mathcal{O}}^r$ and evaluate the candidate by running Level~1.

\paragraph{Agent-priority refinement (Level~2).}
Let $\mathbf{L}_{\mathcal{A}}^r$ be the total ordering over $\mathcal{A}^r$ obtained in initialization by topologically sorting $\bm{\prec}_{\mathcal{A}}^r$.
We select a pending order $o\in\mathcal{O}_{\mathrm{pend}}^r$ and extract the subsequence of its agents from $\mathbf{L}_{\mathcal{A}}^r$.
We then sample $K_{\mathrm{L2}}$ perturbed orderings of these agents, yielding $\tilde{\mathbf{L}}_{\mathcal{A},o}^r=(i_1,\ldots,i_q)$.
For each candidate, we encode it as a chain priority constraint among agents of $o$,
$\tilde{\bm{\prec}}_{\mathcal{A},o}^r:=\{i_l \prec i_{l+1}\mid l=1,\ldots,q-1\}$,
and build $\tilde{\bm{\prec}}_{\mathcal{A}}^r$ by replacing the corresponding part of \cref{eq:precA_pending}
while keeping \cref{eq:precA_inherit,eq:precA_active} unchanged.
Then, we evaluate each candidate by running Level~1.

\section{Pilot Study for Hyperparameter Selection}
\paragraph{Setup.}
We conducted a pilot study to select the hyperparameters for neighborhood search in Level~2 and Level~3.
We used the same protocol as the snapshot evaluations in \cref{sec:evaluation}.
For \texttt{Medium-A}, we fixed $\lambda=1$, $|\mathcal{A}^r|=300$, and $k=5$; for \texttt{Large}, we fixed $\lambda=3$, $|\mathcal{A}^r|=1000$, and $k=6$.
For each setting, we solved 20 instances and report the average normalized makespan (Makespan/LB) over snapshots.
The $k$ values were chosen because the refinement effect of DOPP is most pronounced in these regimes.

\paragraph{Results.}
\Cref{tab:pilot_l2} summarizes sensitivity to $K_{\mathrm{L2}}$ (Level~2), and
\cref{tab:pilot_l3} summarizes sensitivity to $(W_{\mathrm{L3}},K_{\mathrm{L3}})$ (Level~3).
Entries marked ``--'' are omitted: for $W_{\mathrm{L3}}=3$, the window admits only $3!=6$ permutations, so we tested only $K_{\mathrm{L3}}\in\{3,10\}$ and skipped larger values.
These results show that moderate values of $K_{\mathrm{L2}}$ and $K_{\mathrm{L3}}$ perform best.
Larger $K_{\mathrm{L2}}$ and $K_{\mathrm{L3}}$ spend too much time enumerating permutations for a single window, reducing the number of distinct candidates explored within the given refinement time budget.
Based on these results, we set $W_{\mathrm{L3}}=7$ and $K_{\mathrm{L3}}=K_{\mathrm{L2}}=10$ for experiments in \cref{sec:evaluation}.

{
\begin{table}[t]
\centering
\label{tab:pilot_l2}
\setlength{\tabcolsep}{20pt}
\renewcommand{\arraystretch}{1.0}
\begin{tabular}{lcc}
\toprule
$K_{\mathrm{L2}}$ & Medium-A & Large \\
\midrule
3   & 1.076          & \textbf{1.913} \\
10  & \textbf{1.075} & 1.922 \\
50  & 1.104          & 1.937 \\
100 & 1.103          & 1.943 \\
500 & 1.123          & 1.935 \\
\bottomrule
\end{tabular}
\caption{Results of pilot study for Level~2: sensitivity to $K_{\mathrm{L2}}$ (average Makespan/LB over 20 instances; best in bold).}
\end{table}
}
{
\begin{table}[t]
\centering
\label{tab:pilot_l3}
\setlength{\tabcolsep}{5pt}
\renewcommand{\arraystretch}{1.0}
\begin{tabular}{l lccccc}
\toprule
& & \multicolumn{5}{c}{$K_{\mathrm{L3}}$} \\
\cmidrule(lr){3-7}
& $W_{\mathrm{L3}}$ & 3 & 10 & 50 & 100 & 500 \\
\midrule
\multirow{4}{*}{Medium-A} & 3  & 1.07 & 1.06 & --   & --   & -- \\
                          & 7  & 1.06 & \textbf{1.05} & 1.06 & 1.06 & 1.06 \\
                          & 15 & 1.07 & 1.06 & 1.07 & 1.07 & 1.07 \\
                          & 50 & 1.13 & 1.13 & 1.13 & 1.13 & 1.13 \\
\midrule
\multirow{4}{*}{Large}    & 3  & 1.86 & 1.86 & -- & -- & -- \\
                          & 7  & 1.83 & 1.82 & 1.86 & 1.86 & 1.87 \\
                          & 15 & 1.82 & \textbf{1.81} & 1.83 & 1.86 & 1.87 \\
                          & 50 & 1.95 & 1.93 & 1.93 & 1.94 & 1.94 \\
\bottomrule
\end{tabular}
\caption{Results of pilot study for Level~3: sensitivity to $(W_{\mathrm{L3}},K_{\mathrm{L3}})$ (average Makespan/LB over 20 instances; best in bold, ``--'' indicates omitted combinations).}
\end{table}
}
}

\end{document}